\documentclass{article}

\usepackage[preprint,nonatbib]{neurips_2024}

\usepackage[utf8]{inputenc} 
\usepackage[T1]{fontenc}    
\usepackage{hyperref}       
\usepackage{url}            
\usepackage{booktabs}       
\usepackage{amsfonts}       
\usepackage{nicefrac}       
\usepackage{microtype}      
\usepackage{xcolor}         

\usepackage{amsmath}
\usepackage{amssymb}
\usepackage{amsthm}
\usepackage{graphicx}

\usepackage{algorithm}
\usepackage{stmaryrd}  
\usepackage{wrapfig}
\usepackage{multirow}  

\usepackage[noend]{algpseudocode}

\algblockdefx[DoParallel]{DoParallel}{EndDoParallel}[0]{\textbf{Do in parallel}}{}
\makeatletter
\ifthenelse{\equal{\ALG@noend}{t}}%
  {\algtext*{EndDoParallel}}
  {}%
\makeatother

\algrenewcommand\algorithmicrequire{\textbf{Input:}}
\algrenewcommand\algorithmicensure{\textbf{Output:}}
\algnewcommand\Input{\item[\algorithmicrequire]}%
\algnewcommand\Output{\item[\algorithmicensure]}%
\newcommand{\LineComment}[1]{\hfill // \textit{#1}}

\usepackage{cleveref}

\theoremstyle{definition}

\theoremstyle{plain}
\newtheorem{theorem}{Theorem}[section]

\title{PermLLM: Private Inference of Large Language Models within 3 Seconds under WAN}
\author{%
  Fei Zheng, Chaochao Chen, Zhongxuan Han, Xiaolin Zheng\\
  College of Computer Science and Technology, Zhejiang University\\
  \texttt{ \{zfscgy2,zjucc,zxhan,xlzheng\}@zju.edu.cn} \\
}

\begin{document}
\maketitle
\begin{abstract}
    The emergence of ChatGPT marks the arrival of the large language model (LLM) era.
    While LLMs demonstrate their power in a variety of fields, they also raise serious privacy concerns as the users' queries are sent to the model provider.
    On the other side, deploying the LLM on the user's device will also leak all the model data.
    Existing methods based on secure multiparty computation (MPC) managed to protect both the privacy of the model parameters and user queries.
    However, they require gigabytes of data transfer and several minutes to generate just one token, making them impractical for most real-world applications.
    To improve the efficiency of private LLM inference, we propose \textit{PermLLM}, which accelerates the evaluation of non-linear functions using secure random permutation.
    Along with the optimized secret sharing protocols and homomorphic encryption, PermLLM achieves two-party private inference of the ChatGLM-6B model at the speed of around 3s/token, under a realistic network setting (10ms RTT and 1Gbps bandwidth), which is magnitudes faster than existing MPC solutions.
    
\end{abstract}

\section{Introduction}
With the advent of ChatGPT~\cite{chatgpt}, \textit{large language models} (LLMs)~\cite{touvron2023llama,duzhengxiao2022glm,wu2023bloomberggpt} have drawn much attention from both the public and academia.
LLMs have shown great ability in various tasks, e.g., question answering, reading comprehension, text summarization, mathematical reasoning, and so on~\cite{zhao2023llm_survey}.
However, in real-world applications, LLMs still face significant privacy concerns.
For example, consider a typical scenario of ChatGPT's usage where the user sends his query to the OpenAI server, and then gets the response.
In this case, the user's query is exposed to the LLM provider.
This is unacceptable when the query contains sensitive or valuable information, e.g., the user's personal information or confidential data of the company.
Another scenario is that the user downloads the complete LLM model and deploys it on his own device.
Although user privacy is protected in this case, the model's parameters are completely revealed to the user, which violates the model provider's privacy considering LLMs are valuable assets.

Existing approaches to protect both model and data's privacy include \textit{multiparty computation} (MPC) based methods and \textit{split learning}.
\textbf{Multiparty computation}~\cite{secureml2017,mohassel2018aby3} relies on cryptographic primitives.
Although it is theoretically capable of computing any functions, its usage is usually limited to simple machine learning models due to the heavy computation and communication overhead.
To date, current implementations of MPC-based LLMs take at least several minutes and gigabytes to generate one token in ideal settings, i.e., high-performance servers with large bandwidth~\cite{houxiaoyang2023ciphergpt,dongye2023puma,li2022mpcformer}.
The alternative method \textbf{split learning}~\cite{erdogan2022unsplit,vepakomma2018split_health} is also impractical for LLMs since the hidden representation of LLMs could fully reveal the input query~\cite{morris2023embedding_almost,zheng2023attack_vfllm}.

To achieve efficient private inference of LLMs, we propose \textit{PermLLM}, which combines cryptographic technologies with random permutation to realize efficient secure inference of LLMs.
PermLLM adopts a two-party setting with a semi-honest third party like many previous privacy-preserving machine learning studies~\cite{wagh2019securenn,mohassel2018aby3,riazi2018chameleon,knott2021crypten}, 
where one party ($P_0$) is the model provider and the other party ($P_1$) is the user, 
while a semi-honest third party ($P_2$) participates in the \textit{preparation} and \textit{offline} phase to generate pre-computed Beaver's triples and permutation triples.
The crucial idea of PermLLM is to outsource the expensive nonlinear computations to the user, but in a randomly permuted state so that the original data is not revealed.
Although random permutation cannot achieve information-theoretic security, it is secure in the practical sense since the hidden representations in LLM contain thousands of elements, yielding an almost infinite number of possible permutations.
As for linear computations such as matrix multiplication, we adopt the popular additive sharing scheme with the Beaver's triples~\cite{aby2015}, with some improvements regarding the properties of LLM inference, to optimize the performance of online computation.
Finally, the permuted scores are sent to $P_1$, who then obtains the permuted prediction index and performs a computational private retrieval (cPIR) protocol with $P_0$ to obtain the final prediction index.
To demonstrate the efficiency of \textit{PermLLM}, we perform extensive experiments under different settings, and compare it with MPCFormer~\cite{li2022mpcformer} and Puma~\cite{dongye2023puma}, both are state-of-the-art MPC-based secure transformer inference solutions.
Moreover, we provide an implementation of PermLLM based on the ChatGLM-6B model~\cite{2023glm}, an opensource LLM with more than 6 billion parameters.
PermLLM generates each token within seconds under realistic network settings and only consumes about 20Mb of network traffic, while the prediction result remains exactly the same as the original ChatGLM-6B model.
The experiment results indicate that PermLLM is practical for real-world applications.
In summary, we make the following contributions:
\begin{itemize}
    \item 
    We propose PermLLM, which enables fast private inference for LLMs.
    We implement PermLLM on the ChatGLM-6B model, and achieve a token generation speed of 3s/token under a realistic WAN setting, showing its potential for real-world usages.
    \item 
    We propose to compute the nonlinear functions such as GeLU and Softmax on randomly permuted plaintext elements based on the secure secret-shared permutation protocol, avoiding heavy cryptographic operations while the privacy leakage is negligible.
    \item 
    We optimize the secret-shared secure multiplication by leveraging the properties of LLM inference.
    We incorporate it with secure random permutation, homomorphic encryption, and other techniques, achieving a reduction of magnitudes in the computation and communication cost compared with the existing private LLM inference solutions.
\end{itemize}
\section{Related Work}

\textbf{Cryptographic methods.}
There have been many studies applying cryptographic primitives to deep learning in order to achieve data privacy.
For example, CryptoNets~\cite{gilad2016cryptonets} applies fully homomorphic encryption to neural networks inference, and DeepSecure~\cite{rouhani2018deepsecure} represents the neural network by garbled circuits~\cite{yao1986gc} to realize secure inference.
More recently, many hybrid methods have been proposed for more efficient secure neural network computation~\cite{aby2015,secureml2017,riazi2018chameleon,juvekar2018gazelle,wagh2019securenn,rathee2020cryptflow2,knott2021crypten,huangzhicong2022cheetah}.
These methods leverage multiple cryptographic primitives such as homomorphic encryption~\cite{paillier1999,2012bfv1,2012bfv2,ckks2017}, secret sharing~\cite{shamir1979share,beaver1992efficient}, garbled circuits~\cite{yao1986gc}, and probably some customized MPC protocols~\cite{wagh2019securenn}.
While these methods could be practical in certain applications, their applications to LLMs are difficult due to the large model size.
Currently, secure LLMs based on cryptographic methods~\cite{li2022mpcformer,houxiaoyang2023ciphergpt,dongye2023puma} take at least several minutes and GBs of communication to generate just one token, even under extremely ideal network environments.
Such cost is currently impractical for real-world LLM applications, and it is expected that achieving practicality with existing cryptographic tools will be challenging~\cite{houxiaoyang2023ciphergpt}.
The major efficiency bottleneck for cryptographic methods is the nonlinear computation.
To improve the efficiency of nonlinear computation, many methods have been proposed like garbled circuits~\cite{secureml2017,juvekar2018gazelle}, GMW protocol~\cite{knott2021crypten,riazi2018chameleon}, segmented approximation~\cite{secureml2017,mohassel2018aby3}, and polynomial approximation~\cite{haomeng2022iron_transformers,houxiaoyang2023ciphergpt,dongye2023puma}.
However, those methods are still quite slow compared with linear computations (addition/multiplication), and nonlinear computation remains the most consuming part in many privacy-preserving machine learning systems.

\noindent\textbf{Random permutation.} 
A walkaround for the cryptographically expensive nonlinear computation is to evaluate the randomly permuted plaintext~\cite{zheng2022permute,liu2024pp-stream}.
Although random permutation cannot fully protect data privacy as the set of elements is exposed, it is shown to be secure in a practical sense when the number of elements is not too small~\cite{zheng2022permute,liu2024pp-stream}.
The permutation invariance in transformers is also considered
~\cite{xu2023shuffled_transformer}.
Although claimed to be privacy-preserving, it is obvious that the model utility is preserved in this way.

\section{Preliminaries}
\subsection{Structure of LLM}
\begin{wrapfigure}{r}{0.4\textwidth}
\vspace{-15pt}
\begin{center}
    \includegraphics[width=1\linewidth]{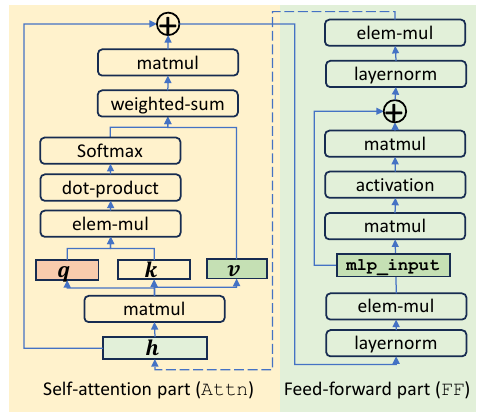}
    \vspace{-15pt}
    \caption{Transformer layer structure.}
    \label{fig:transformer_layer}
\end{center}
\vspace{-5pt}
\end{wrapfigure}

A typical LLM can be represented as follows:
\begin{equation}
    \mathtt{LLM}(X) = L_\text{head} \circ T_L \circ \cdots T_1 \circ \mathtt{Emb}(X),
\end{equation}
where \texttt{Emb} is the word embedding layer, $T_i$s are the transformer layers, and $L_\text{head}$ is the last dense layer for next token prediction.
\texttt{Emb} and $L_\text{head}$ share the same weight matrix.

While \texttt{Emb} and $L_\text{head}$ are just dense layers that can be viewed as matrix multiplications,
the transformer layers are complicated and contain the most computations and parameters of the LLM.
The transformer layer has two parts, i.e., the self-attention module and the feed-forward module.
The self-attention module consists of several dense layers and the attention layer, while the feed-forward module also consists of several dense layers, with an activation layer and two layer normalization layers.
\Cref{fig:transformer_layer} shows the architecture of a typical transformer layer.

\subsection{Additive Secret Sharing}
\label{sec:a-ss}
Additive secret sharing (A-SS)~\cite{shamir1979share} is a widely used two-party secure computation method for arithmetic circuits.
We say a value $x$ is (additively) shared between two parties ($P_0$ and $P_1$) when $P_0$ holds $\langle x \rangle_0$ and $P_1$ holds $\langle x \rangle_0$, such that $\langle x \rangle_0 + \langle x \rangle_1 = x$.
$\langle x \rangle_i$ is then called the share of $P_i$.
We use $\langle x \rangle$ to denote that value $x$ is in a shared manner.
To add two shared values, both parties simply add their shares.
The multiplication of shared values is slightly more complicated.
It is often done using the Beaver's triples~\cite{beaver1992efficient}.
Suppose $P_0$ and $P_1$ hold two shared values $x$ and $y$, and want to compute their product $xy$.
The procedure of multiplication can be described as follows.
\begin{enumerate}
    \item (Offline) $P_0$ and $P_1$ first obtain the additively shared Beaver's triples, i.e., $u, v, w$ such that $uv = w$ ($u, v$ have the same shape with $x, y$).
    \item $P_0$ and $P_1$ reconstruct $x - u$ and $y - v$  by revealing their shares $\langle x-u \rangle_i$ and $\langle y - v\rangle_i$.
    \item $P_0$ computes $\langle z \rangle_0 = (x - u)(y - v) + \langle x \rangle_0 (y - v) + (x - u) \langle v \rangle_0 + \langle w \rangle_0$;
    $P_1$ computes $\langle z \rangle_1 = \langle x \rangle_1 (y - v) + (x - u) \langle v \rangle_1 + \langle w \rangle_1$.
    Thus, $\langle z\rangle_i$ is $P_i$'s share of $xy$.
\end{enumerate}
To achieve information-theoretic security, A-SS is defined on a finite domain such as the integer ring.
In this case, a single party's view is always uniformly random and irrelevant to the plaintext values.
\section{Secure Building Blocks}
From the architecture of LLM, we can see that the computation in LLM inference mainly includes matrix multiplications in dense layers and attention layers, the argmax operation (or other sampling strategies) used to generate the next token, and nonlinear functions including the activation function, layer normalization, and Softmax.
Correspondingly, in this section, we describe the building blocks for PermLLM, including matrix multiplication protocols based on secret sharing, next token prediction based on the BFV cryptosystem, and our proposed nonlinear evaluation protocol based on secure random permutation.

\subsection{Security Setting}
PermLLM targets the case of two-party LLM inference.
We denote the model provider as $P_0$ and the model user as $P_1$.
For simplicity, we also assume there is a third party $P_2$ who will assist $P_0$ and $P_1$ in the precomputation during the offline phase.
We assume that each party is semi-honest, i.e., they will not deviate from the protocol, but they will exploit any information they receive during the execution of the protocol, which is a common setting used in previous privacy-preserving machine learning studies~\cite{wagh2019securenn,mohassel2018aby3,li2022mpcformer}.

\noindent\textbf{Execution phases.}
The execution of PermLLM is divided into three phases, i.e., \textit{preparation}, \textit{offline}, and \textit{online}.
The online phase means the actual secure computation execution when the input is delivered to the corresponding party, i.e. $P_1$ obtains the input token.
Before the execution of the online phase, an offline phase is consumed to distribute the pre-computed values.
The difference between the offline phase and the online phase is that the offline phase does not require the actual value of the inputs, so it can be executed ahead of the actual task.
The preparation phase will only be executed once unless the model parameters of the LLM are modified.

\subsection{Linear Computations}
The linear computations in LLM inference can be divided into two types:
\begin{itemize}
    \item \textbf{One operand fixed case}: During the inference, the model weights are all fixed. 
    Thus, the computation of dense layers and the embedding retrieval layer (where the word embedding is multiplied with the one-hot vector representing the input token) belong to this type.
    \item \textbf{One operand growing case}: 
    The text generation task requires multiple inference steps to generate a sequence of output tokens.
    Therefore, by leveraging the property of the attention mechanism, the past key/value vectors in each layer can be cached for future computation.
    In other words, in the linear computations of the attention mechanism, one operand keeps `growing', i.e., a new part will be appended to the existing operand during each inference.
\end{itemize}
For each type, we design a modified secure multiplication protocol to minimize the computation cost.
The insight is that the interactive computation of the `fixed part' of the operand only needs to be executed once.
Suppose the fixed or growing operand is $x$, and another operand is $y$.
Recall the secure multiplication procedure introduced in \Cref{sec:a-ss}.
If $x$ is fixed throughout the computation, the process of $P_2$ generating $u$ and $P_0, P_1$ reconstructing $x - u$ can be executed \textit{only once} in the preparation phase.
Similarly, if $x$ grows during each inference, $P_2$ only needs to compute $u$ for the \textit{newly appended part} of $x$, instead of the \textit{complete} $x$.
We formally describe the secure multiplication protocols of those two cases in \Cref{alg:secure_mul_fixed,alg:secure_mul_growing}.
Notice the multiplication here can be scalar multiplication, matrix multiplication, or any operation satisfying the distributive property.

\begin{algorithm}[h]
\small
\caption{\textsf{SecureMul}$_F$}
\label{alg:secure_mul_fixed}
\begin{algorithmic}[1]
\Input $P_0$ holds constant $X$. Shared value $\langle Y \rangle$.

\Output Shared product $\langle Z\rangle = \langle XY \rangle$.

\item[\underline{Preparation:}]
\State $P_2$ generates random $U$ (which has the same shape as $X$) and sends it to $P_0$.
\State $P_0$ sends $X - U$ to $P_1$.

\item[\underline{Offline:}]
\State $P_2$ generates random $V$ (which has the same shape as $Y$), and shares $V, W \gets UV$ to $P_0$ and $P_1$.

\item[\underline{Online:}]
\State $P_0$ and $P_1$ reconstructs $Y - V$.
\State $P_0$ computes $\langle Z \rangle_0 \gets X(Y-V) + (X-U)\langle V \rangle_0 + \langle W \rangle_0$.
\State $P_1$ computes $\langle Z \rangle_1 \gets (X-U) \langle V \rangle_1 + \langle W \rangle_1$.
\end{algorithmic}
\end{algorithm}
\vspace{-5pt}

\begin{algorithm}[h]
\small
\caption{\textsf{SecureMul}$_G$}
\label{alg:secure_mul_growing}
\begin{algorithmic}[1]
\Input Shared value $\langle X' \rangle$, $\langle X \rangle$ and $\langle Y \rangle$.
During each online execution, $X$ is updated by appending $X'$ to it, i.e., $X \gets \mathsf{concat}(X, X')$.

\Output Shared product $\langle Z \rangle =\langle XY \rangle$.

\item[\underline{Preparation:}]
\State 
    $P_0$ sets $X - U \gets null, \langle U \rangle_0 \gets null$.
    $P_1$ sets $X - U \gets null, \langle U \rangle_1 \gets null$.
\State $P_2$ sets $U \gets null$.

\item[\underline{Offline:}]
\LineComment{We denote $\mathsf{concat}(null, U) = U$}
\State $P_2$ generates random $U', V$ (which has the same shape as $X', Y$), and sets $U\gets \mathsf{concat}(U, U'), W \gets UV$.

\State $P_2$ shares $U', V, W$ to $P_0$ and $P_1$.

\item[\underline{Online:}]
\State $P_0$ and $P_1$ reconstruct $X' - U'$ and $Y - V$.
\State $P_0$ and $P_1$ set $X - U \gets \mathsf{concat}(X - U, X'-U')$.

\State 
    $P_0$ sets $\langle U \rangle_0 \gets \mathsf{concat}(\langle U \rangle_0, \langle U' \rangle_0$).
    $P_1$ sets $\langle U \rangle_1 \gets \mathsf{concat}(\langle U \rangle_1, \langle U' \rangle_1$).
\State $P_0$ computes $\langle Z \rangle_0 \gets (X-U)(Y-V) + (X-U)\langle V \rangle_0 + \langle U\rangle_0(Y-V) + \langle W \rangle_0$.
\State $P_1$ computes $\langle Z \rangle_1 \gets (X-U) \langle V \rangle_1 + \langle U\rangle_1(Y-V) + \langle W \rangle_1$.
\end{algorithmic}
\end{algorithm}

\subsection{Nonlinear Computation}
The key idea in PermLLM is to evaluate the non-linear functions using random permutation, based on the observation that shuffled plaintext elements are quite safe to reveal.
Different from the method proposed in \cite{zheng2022permute} where the third party ($P_2$) has to participate in the online phase, we design a different protocol based on a secret-shared permutation protocol~\cite{chase2020shuffle}, where $P_1$ obtains the shuffled plaintext.
We only use the online part of the protocol in~\cite{chase2020shuffle}, and distribute the offline computation to $P_2$ to reduce the cost.

Suppose $P_0$ and $P_1$ want to compute the shares of $f(\mathbf x)$, where $f$ is an element-wise function and $\mathbf x$ is a (shared) vector.
First, $P_1$ initializes a random permutation $\pi$, and $P_0$ and $P_1$ collaboratively perform the secure permutation protocol, which reconstructs $\pi[\mathbf x]$ on $P_1$.
$P_1$ then compute $f(\pi[\mathbf x])$ in plaintext.
After this, $P_0$ and $P_1$ again perform the secure permutation protocol, using the inverse permutation $\pi^{-1}$, to obtain the shared values of $f(\mathbf x) = \pi^{-1}[f(\pi[\mathbf x])]$.
We describe the secure permutation protocol and the nonlinear computation protocol in \Cref{alg:secure_perm,alg:nonlinear} respectively.

\begin{wrapfigure}{r}{0.48\textwidth}
\vspace{-15pt}
\begin{minipage}{\linewidth}
\begin{algorithm}[H]
\small
\caption{\textsf{SecurePerm}}
\label{alg:secure_perm}
\begin{algorithmic}[1]
\Input $P_0$ holds a permutation $\pi$. Shared value $\langle \mathbf x \rangle$.
\Output Shared permuted value $\langle \mathbf y \rangle  = \langle \pi[\mathbf x] \rangle$.
\item[\underline{Offline:}]
\State $P_0$ sends $\pi$ to $P_2$.
\State $P_2$ randomly generates $\mathbf r_1, \mathbf r_2$, and sets $\Delta \gets \pi[\mathbf r_0] - \mathbf r_1$.
\State $P_2$ sends $\Delta$ to $P_0$, and $\mathbf r_0, \mathbf r_1$ to $P_1$.
\item[\underline{Online:}]
\State $P_1$ sends $\langle x \rangle_1 - \mathbf r_0$ to $P_0$, and sets $\langle y \rangle_1 = \mathbf r_1$.
\State $P_0$ sets $\langle y \rangle_1 = \pi[\langle x \rangle_0] + \pi[\langle x \rangle_1 - \mathbf r_0] + \Delta$.
\end{algorithmic}
\end{algorithm}
\end{minipage}
\vspace{-5pt}
\end{wrapfigure}

\begin{algorithm}[h]
\small
\caption{\textsf{SecureNonlinear}}
\label{alg:nonlinear}
\begin{algorithmic}[1]
\Input $P_0$ holds a permutation $\pi$. 
    Shared value $\langle \mathbf x \rangle$. 
    A public nonlinear function $f$ such that $f(\pi[\mathbf x]) = \pi[f(\mathbf x)]$ for any permutation $\pi$.
\Output Shared output $\langle \mathbf y \rangle = \langle f(\mathbf x) \rangle$.
\State $P_0$ and $P_1$ execute $\langle \pi[\mathbf x] \rangle \gets \mathsf{SecurePerm}(\pi, \mathbf x)$, and reconstruct $\pi[\mathbf x]$ on $P_1$.
\State $P_1$ computes $\langle \mathbf y' \rangle_1 \gets f(\pi[\mathbf x])$. $P_0$ sets $\langle \mathbf y' \rangle \gets \mathbf 0$.
\State $P_0$ and $P_1$ execute $\langle \mathbf y \rangle \gets \mathsf{SecurePerm}(\pi^{-1}, \langle \mathbf y' \rangle)$.
\end{algorithmic}
\end{algorithm}

\noindent\textbf{2D Permutation.}
In the above \textsf{SecureNonlinear} protocol, we assume that the non-linear function is element-wise, so the input will be flattened into a 1D vector before executing  \textsf{SecureNonlinear}.
However, nonlinear functions encountered in LLM are not usually completely element-wise, such as Softmax and layer normalization.
Instead, they can only be viewed element-wise in the last dimension.
In this case, we can split an input tensor $X \in \mathbb R^{n\times d}$ into $n$ vectors $(\mathbf x_1, \cdots, \mathbf x_n)$.
In this way, we have $\pi^{-1}[f(\pi[\mathbf x_i])] = f(\mathbf x_i)$ for any $d$-permutation $\pi$.
Thus, for such nonlinear functions and input $X\in\mathbb R^{n \times d}$, we define 2D permutation as follows: 
First, permute the $d$ elements inside each row using distinct random permutations, and then permute the $n$ rows.
The inverse 2D permutation can also be easily obtained by first inverse the row-level permutation and then the element-level permutations.
In this way, we can extend the \textsf{SecureNonlinear} protocol to Softmax, layer normalization, and other similar functions.

\subsection{Next Token Retrieval}
The output of LLM is an $N$-dimensional vector representing the scores/probabilities on the vocabulary.
There are many strategies to generate the next token from the scores, e.g., the greedy strategy selects the next token with the largest scores, and sampling-based strategies select the next token by chance based on the score.
However, the score vector is a projection of the word embedding, i.e., $\mathbf s = E \mathbf h$, where $E$ is the word embedding, and $\mathbf h$ is the last hidden representation.
If $P_1$ collects multiple score vectors, he also obtains the information of the word embedding, which is considered privacy leakage.

To protect the word embedding, PermLLM does not directly reconstruct the scores on $P_1$.
Instead, $P_1$ receives a randomly permuted score vector $\mathbf s' = \pi[\mathbf s]$.
Based on the permuted scores, $P_1$ chooses the permuted index of the next token according to his own strategy.
After that, $P_1$ and $P_0$ perform PIR (Private Information Retrieval) to get the actual index.

In PermLLM, we adopt the cPIR (computational private retrieval) protocol~\cite{xunyi2013pir,melchor2016xpir}, 
and use the BFV homomorphic encryption~\cite{2012bfv1,2012bfv2}, because it is efficient and works on the integer ring.
The idea is to let $P_1$ encrypt the permuted index as a one-hot integer, then send it to $P_0$ who performs a homomorphic dot-product with the inverse permutation vector $\mathbf p = (\pi^{-1}(1), \cdots, \pi^{-1}(n))$.
Because the BFV cryptosystem naturally supports data packing (under the common 128-bit security setting we can pack 2048 values into one ciphertext), and the vocabulary size is around 130K, we reduce the number of cryptographic operations to around 65.

We describe the algorithm (take argmax as the example) formally in \Cref{alg:prediction}, where $N\approx 130$K is the vocabulary size, $L=2048$ is the number of slots in the ciphertext, \textsf{Enc} and \textsf{Dec} are the encryption and decryption function.
We use $\llbracket \cdot \rrbracket$ to denote the ciphertext, and $\oplus, \odot$ to denote the homomorphic addition and multiplication.
\begin{algorithm}[h]
\small
\caption{\textsf{SecurePrediction} (argmax)}
\label{alg:prediction}
\begin{algorithmic}[1]
\Input Shared value $\langle \mathbf s \rangle$ (score vector).
\Output $P_0$ holds $i = \mathop{\text{argmax}}_i \mathbf s[i]$.
\State $P_0$ and $P_1$ execute $\langle \pi[\mathbf s] \rangle \gets \mathsf{SecurePerm}(\pi, \langle \mathbf s \rangle)$ and reconstruct $\pi[\mathbf s]$ on $P_1$.
\LineComment{Only $P_0$ has $\pi$}
\State $P_1$ computes $\pi[i] \gets \mathop{\text{argmax}}_{i'} \pi[\mathbf s]$.
\State $P_1$ sets $\mathbf p \gets \mathsf{onehot}(\pi[i])$, then divides $\mathbf p$ into $\lceil N/L \rceil$ pieces of length $L$, i.e, $\mathbf p_1, \cdots, \mathbf p_{\lceil N/L \rceil}$.
\State $P_1$ computes $\llbracket p_j \rrbracket \gets \mathsf{Enc}(\mathbf p_j), j  = 1..{\lceil N/L \rceil}$, and sends them to $P_0$.
\State $P_0$ also divides the vector $(\pi^{-1}[1], \cdots, \pi^{-1}[N])$ into $\lceil N/L \rceil$ pieces $\mathbf q_1, \cdots, \mathbf q_{\lceil N/L \rceil}$, and homomorphically computes $\llbracket i \rrbracket \gets (\llbracket \mathbf p_1 \rrbracket \odot \mathbf q_1) \oplus \cdots \oplus (\llbracket \mathbf p_{\lceil N/L \rceil} \rrbracket \odot \mathbf q_{\lceil N/L \rceil})$.
$P_0$ then sends $\llbracket i \rrbracket$ to $P_1$.
\State $P_1$ gets $i \gets \mathsf{Dec}(\llbracket i \rrbracket)$.
\end{algorithmic}
\end{algorithm}
\vspace{-8pt} 
\section{Optimizing PermLLM}
In this section, we describe the optimizations on the implementation of PermLLM, including making some unimportant LLM parameters public, and the secret sharing on real number for GPU computation.
\subsection{Making Some Parameters Public}
We notice that the multiplication between a secret-shared value and a public value is almost cost-free.
Consider a public value $a$ and a shared value $\langle x \rangle$.
To perform multiplication, each party locally computes $a\langle x\rangle_i$, and they obtain the shares of $ax$.
Hence, if some of the LLM parameters are made public, the computation cost of secret sharing-based protocols could be reduced.
Hence, in PermLLM, we reveal the \textit{element-wise affine weight} in layer normalization modules.
The affine weights only make up to $< 0.01\%$ of total model parameters, and are usually close to $\mathbf 1$.
%
%
Hence, we assume they contain less useful information and can be revealed securely.
Also, the \textit{positional embeddings} in LLMs are generated by some common rules, hence we consider them as public parameters. 

We consider the other parts of the LLM parameters important and shall be kept private.
For example, word embeddings are useful and can be used to fine-tune the other language models.
Also, the weights in attention layers are critical to the performance and functionality of LLMs, e.g., a low-rank change on them (LoRA) significantly influences the performance on certain tasks~\cite{2022lora}.
However, the relationship between the hardness of stealing the model and the public parameters remains to be studied.
One can also choose to make more parameters public as a tradeoff between efficiency and model privacy.

\subsection{Secret Sharing on Real Numbers}
\label{sec:impl-ss}
Secret sharing is originally designed for integer rings, in which it achieves information-theoretical security since each party only receives uniformly random views.
However, modern GPUs are designed for float-point computation and are rather insufficient for integer operations.
Hence, we extend the additive sharing scheme to real numbers for efficiency and convenience~\cite{tjell2021real_sharing,gundersen2023real_sharing}.
Suppose $P_0$ and $P_1$ are performing secure multiplication between $x$ and $y$, with the beaver triples $(u, v, w)$.
Recall that $x - u$ and $y - v$ are revealed during the multiplication.
In this case, we have to assign $u$ a much larger scale than $x$, so that $x - u$ exposes little information about $x$ (similar for $v$).
This is done by choosing $\mathbb E[u^2] = K^2\mathbb E[x^2]$, where $K$ is a coefficient to control the scale ratio.
Larger $K$ leads to a higher security level due to the larger noise/signal ratio.

Suppose $\mathbb E[x^2], \mathbb E[y^2]$ is known, the generation of beaver triples can be expressed as follows:
Randomly select $\langle u \rangle_0, \langle u\rangle_1 \sim U, \langle v\rangle_0, \langle v\rangle_1 \sim V, \langle w\rangle_0 \sim W$ independently such that $\mathbb E[U^2] = K^2 \mathbb E[x^2]$, $\mathbb E[V^2] = K^2\mathbb E[y^2]$, and $\mathbb E[W^2] = K^2\mathbb E[x^2y^2]$, $\langle w\rangle_1 = (\langle u\rangle_0 + \langle u\rangle_1)(\langle v\rangle_0 + \langle v\rangle_1) - \langle w\rangle_0$.
Similarly, when sharing a plaintext value $x$, one also choose $\langle x \rangle_0 \sim R$ such that $\mathbb E[R^2] = K^2\mathbb E[x^2]$, and $\langle x \rangle_1 = x - \langle x \rangle_0$.

In this way, we can ensure that any share/revealed value will be either a random value or a large random `noise' plus the original value.
We will prove that such property will be preserved after our secure multiplication protocols in the next section.
Hence, although information-theoretic security cannot be achieved, we can still obtain a high security level by controlling the scale of noise.

It is also worth noting that real number secret sharing is not necessary for PermLLM, it is just an optimization regarding the modern GPU's computation power for floating-point operations.
One can always turn it into integer secret sharing.

\section{Security Analysis}
Although PermLLM cannot achieve provable security in an information-theoretic sense, we show that PermLLM actually leaks very little information from both parties.
In this section, we discuss the privacy leakage from two aspects, i.e., the random permutation method and the floating-point secret sharing.

\subsection{Random Permutation}
In PermLLM, the random permutation is applied to the attention scores, hidden embeddings, and the score vector for the prediction.
For attention scores, the random permutation is performed on a matrix of shape $h \times n$, representing the attention scores of $h$ attention heads and $n$ key vectors ($n$ is the length of input).
The permutation applied here is a 2D permutation, where we permute the $h$ rows first, and in each row, we permute the $n$ elements.
This results in $h!(n!)^h$ possible permutations.
In LLM models, $h \ge 32$, so we have $h!(n!)^h > 32! > 2^{117}$.
As for the other cases, the number of elements is even larger (e.g., $\ge 4096$), so we have an even larger number of permutations ($4096! > 2^{43250}$).
We can conclude that the probability of guessing the correct permutation in such cases is negligible.
Previous studies also show that permuted vector are almost irrelevant of the original vector in the statistical sense~\cite{zheng2022permute,liu2024pp-stream}.

\textbf{Brute-force attack.}
Since the random permutation reserves the set of elements, it is theoretically possible for the adversary with exponential computation power to enumerate all possible inputs to match the set of permuted elements (under the assumption that the sets of hidden embedding elements of different inputs are distinguishable).
Although this kind of attack could be a valid concern under extreme conditions, PermLLM naturally avoids such attacks as the permuted elements are revealed on $P_1$ (the user) instead of $P_0$ (the model provider).

\subsection{Real Number Secret Sharing}
Unlike on the finite rings (e.g., $\mathbb Z_N$), secret sharing cannot achieve information-theoretic security on real numbers.
However, as we described in \Cref{sec:impl-ss}, we control the security level by choosing the noise scale.
In this way, we can ensure that $\mathbb E[(\langle x \rangle_i -x)^2] \ge K^2\mathbb E[x^2]$ after sharing, and $\mathbb E[(x - u)^2] \ge K^2\mathbb E[x^2]$ during the secure multiplication.
In other words, the revealed value is the original value plus a noise $K$ times larger in magnitude.
Here, we prove that such property preserves during our secure multiplication protocols.

\begin{theorem}[Noise scale after multiplication]
    If $\mathbb E[(\langle x\rangle_i - x)^2] \ge K^2 \mathbb E[x^2]$, and $\mathbb E[(\langle y \rangle_i - y)^2] \ge K^2 \mathbb E[y^2]$ for $i = 0, 1$, and $\langle z\rangle = \mathsf{SecureMul}_F(x, y)$ or $\mathsf{SecureMul}_G(x, y)$, 
    where the beaver triples is generated in the way described in \Cref{sec:impl-ss}.
    Then we have $\mathbb E[(\langle z\rangle_i - z)^2] \ge K^2 \mathbb E[z^2]$.
\end{theorem}
\begin{proof}
    $\mathsf{SecureMul}_F$:
    First, $(\langle z \rangle_0 - z)^2 = (-x \langle v \rangle_1 - u\langle v \rangle_0 + \langle w \rangle_0)^2$, and $\langle w \rangle_0$ is independent with the other terms, hence
    $\mathbb E[(\langle z \rangle_0 - z)^2] \ge \mathbb E[\langle w_0 \rangle^2] = K^2\mathbb E[(xy)^2]$.
    Second, $\mathbb E[(\langle z\rangle_1 - z)^2] = \mathbb E[(-xy + x\langle v\rangle_1 + u\langle v\rangle_0 - \langle w \rangle_0)^2] \ge \mathbb E[\langle w\rangle_0^2] = K^2\mathbb E[(xy)^2]$.
    Similarly, in $\mathsf{SecureMul}_G$ we can also extract the independent term $\langle w\rangle_0$ and obtain the same result.
\end{proof}

Hence, in real number secret sharing, we can always control the security level by the selection of $K$.
In practice, we choose $K = 100$, meaning that the noise is 100 times larger than the original value.
This sharing scheme introduces some rounding errors during the computation, but the error is negligible using the common float32 data format.
Also note that in PermLLM, we do not re-share any value, hence the attacker cannot collect multiple shares to infer the original value.
\section{Experiments}
To evaluate the performance of PermLLM, we conduct multiple experiments under different settings, i.e., model structure and network environment.
We compare PermLLM with the state-of-the-art MPC solution for LLM inference, i.e., MPCFormer~\cite{li2022mpcformer} and Puma~\cite{dongye2023puma}.
The experiments are performed on a server with 4 NVIDIA RTX 3090 (24G) GPUs, except for the realization of ChatGLM-6B, which is performed on a server with 3 NVIDIA L20 (48G) GPUs.
We use PyTorch~\cite{paszke2019pytorch} for floating-point computations and Pyfhel~\cite{ibarrondo2021pyfhel} for the BFV cryptosystem.

\textbf{Network settings.}
We perform experiments under LAN (localhost) and two WAN (wide area network) settings.
For two WAN settings, 
we set the RTT (Round Trip Time) to 10ms and 20ms,
corresponding to the network connection between two relatively close cities, e.g., Tokyo-Osaka (400km), and the network connection between two cities of moderate distance, e.g., Berlin-London (930km)
\footnote{Data is obtained from wondernetwork.com/ping}.
The bandwidth is set to 1Gbps and 100Mbps, corresponding to the commercial high bandwidth and the normal bandwidth for family usage.
%

\subsection{Benchmarks}
To demonstrate the effectiveness of PermLLM, we perform benchmarks on nonlinear operations and transformer layer inference.
We measure the communication size, number of communication rounds, and running time under different network settings.

The results for nonlinear operations are reported in \Cref{tab:nonlinear}.
PermLLM achieves the best efficiency in all benchmarks.
This is because PermLLM avoids heavy cryptographic operations by exposing the randomly permuted plaintexts, so that the cost is almost irrelevant to the nonlinear function itself (as local computation cost is negligible).
The cost is mainly made of one secure permutation plus one homomorphical dot-product (for argmax), or two secure permutations (for element-wise functions).
However, for MPCFormer or Puma, the nonlinear functions are approximated by high-order polynomials or computed by boolean circuits, which is expensive for secure computation.

The results for single transformer layer inference are reported in \Cref{tab:layer}.
We can see that PermLLM exhibits least communication size and running time under all settings.
For the small transformer ($d=768$), PermLLM is at least one magnitude faster than others, and for the large transformer ($d=4096$, used in LLMs), the advantage becomes around two magnitudes.

We also notice that in MPCFormer, the secret shared multiplication based on the Crypten~\cite{knott2021crypten} is unoptimized for the fixed weight matrix, resulting in a huge additional overhead.

\begin{table}[]
\renewcommand\arraystretch{0}
\small
\caption{Benchmark results on nonlinear operations. The right 3 columns are time consumption (seconds) under different network settings. `LN' is laryer normalization.}
\label{tab:nonlinear}
\begin{tabular}{@{}cccrrrrr@{}}
\toprule
$f$ & size & Method & Comm.(Mb) & Rounds & LAN & 1Gbps/10ms & 100Mbps/20ms \\ \midrule
\multirow{6}{*}{\parbox[t]{2mm}{\multirow{3}{*}{\rotatebox[origin=c]{90}{Argmax}}}} & \multirow{3}{*}{1K} & MPCFormer & 1.24 & 101 & 0.09 (0.01) & 2.11 (0.00) & 4.12 (0.01) \\ \cmidrule(l){3-8} 
 &  & Puma & $\approx$0.80 & $\approx$1700 & 0.15 (0.03) & 3.29 (0.03) & 6.33 (0.05) \\ \cmidrule(l){3-8} 
 &  & \textbf{PermLLM} & \textbf{0.26} & \textbf{4} & \textbf{0.03 (0.00)} & \textbf{0.05 (0.00)} & \textbf{0.06 (0.01)} \\ \cmidrule(l){2-8} 
 & \multirow{3}{*}{100K} & MPCFormer & 110.9 & 153 & 0.78 (0.04) & 4.04 (1.41) & 10.74 (0.50) \\ \cmidrule(l){3-8} 
 &  & Puma & $\approx$80 & $\approx$2500 & 1.17 (1.13) & 6.23 (0.08) & 12.88 (0.06) \\ \cmidrule(l){3-8} 
 &  & \textbf{PermLLM} & \textbf{4.02} & \textbf{4} & \textbf{0.12 (0.01)} & \textbf{0.17 (0.01)} & \textbf{0.2 (0.00)} \\ \midrule
\multirow{6}{*}{\parbox[t]{2mm}{\multirow{3}{*}{\rotatebox[origin=c]{90}{Element-wise}}}} & \multirow{3}{*}{1K} & Puma: Gelu & $\approx$1.7 & $\approx$200 & 0.04 (0.02) & 0.38 (0.02) & 0.71 (0.01) \\ \cmidrule(l){3-8} 
 &  & Puma: LN & $\approx$0.20 & $\approx$190 & 0.02 (0.00) & 0.35 (0.01) & 0.67 (0.01) \\ \cmidrule(l){3-8} 
 &  & \textbf{PermLLM: any} & \textbf{0.01} & \textbf{3} & \textbf{0.00 (0.00)} & \textbf{0.02 (0.00)} & \textbf{0.04 (0.01)} \\ \cmidrule(l){2-8} 
 & \multirow{3}{*}{\begin{tabular}[c]{@{}c@{}}100K\\100$\times$1K\end{tabular}} & Puma: Gelu & $\approx$170 & $\approx$200 & 1.32 (0.11) & 1.69 (0.23) & 12.37 (0.08) \\ \cmidrule(l){3-8} 
 &  & Puma: LN & $\approx$20 & $\approx$200 & 0.13 (0.03) & 0.15 (0.02) & 0.83 (0.03) \\ \cmidrule(l){3-8} 
 &  & \textbf{PermLLM: any} & \textbf{1.15} & \textbf{3} & \textbf{0.01 (0.00)} & \textbf{0.03 (0.01)} & \textbf{0.05 (0.00)} \\ \bottomrule
\end{tabular}
\end{table}

\begin{table}[h]
\vspace{-5px}
\renewcommand\arraystretch{0}
\small
\centering
\caption{Benchmark results on single transformer layer inference.}
\label{tab:layer}
\begin{tabular}{@{}ccrrrrr@{}}
\toprule
Model size & Method           & Comm.(Mb)    & Rounds        & LAN                   & 1Gbps/10ms            & 100Mbps/20ms         \\ \midrule
\multirow{3}{*}{\begin{tabular}[c]{@{}c@{}}Large\\ ($d=4096$)\end{tabular}} & MPCFormer & 3073.79 & 53 & 1.37 ± 0.01 & 26.15 ± 0.35 & 259.65 ± 0.36 \\ \cmidrule(l){2-7} 
           & Puma        & $\approx$33   & $\approx$1500 & 8.82 (0.09)           & 11.45 ± 0.42         & 14.40 ± 0.65         \\ \cmidrule(l){2-7} 
           & \textbf{PermLLM} & \textbf{0.49} & \textbf{20}   & \textbf{0.04 ± 0.01}  & \textbf{0.12 ± 0.01} & \textbf{0.24 ± 0.00} \\ \midrule
\multirow{3}{*}{\begin{tabular}[c]{@{}c@{}}Small\\ ($d=768$)\end{tabular}}  & MPCFormer & 108.32  & 53 & 1.29 ± 0.19 & 1.71 ± 0.01  & 10.29 ± 0.20  \\ \cmidrule(l){2-7} 
           & Puma         & $\approx$6    & $\approx$1000 & 0.46 (0.03)           & 2.24 ± 0.03          & 3.97 ± 0.04          \\ \cmidrule(l){2-7} 
           & \textbf{PermLLM} & \textbf{0.1}  & \textbf{20}   & \textbf{0.034 ± 0.00} & \textbf{0.15 ± 0.01} & \textbf{0.23 ± 0.02} \\ \bottomrule
\end{tabular}
\end{table}

\subsection{ChatGLM-6B benchmark}
\begin{figure}[h!]
    \centering
    \includegraphics[width=1\linewidth]{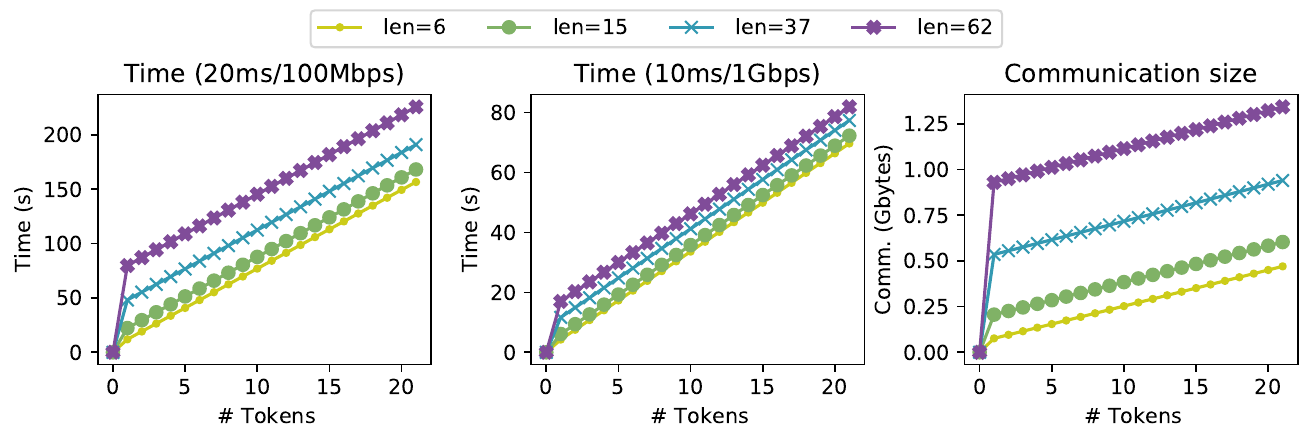}
    \vspace{-5px}
    \caption{Time consumption and communication size of the ChatGLM-6B private inference.}
    \label{fig:enter-label}
\end{figure}

We implement the ChatGLM-6B model using PermLLM, and test it under the two simulated WAN settings.
The prompt length is chosen to be 6, 15, 37, and 62, and the generation length is set to 20.
We plot the time consumption of the generation process for each generated token.
Under the network of 20ms/100Mbps, the token generation speed is around 7s/token except for the first token because the whole prompt is forwarded.
For a faster network (10ms/1Gbps), the speed becomes around 3s/token, which can be considered practical for real-world usage.
We observe that the generation speed remains constant regardless of the number of previously generated tokens, because we have optimized secure multiplication protocols and do not perform repeated interactive computation for past key/value vectors.

As for the accuracy loss, since PermLLM only introduces negligible rounding error compared to the plaintext model, the output is identical to the original ChatGLM-6B model during our experiments.
\section{Conclusion}
In this paper, we introduce PermLLM, a fast private inference framework for large language models.
By leveraging the random permutation based nonlinear evaluation along with optimizing the secret sharing based protocols, PermLLM greatly improves the efficiency compared to other methods, at a negligible privacy cost.
PermLLM achieves a token generation speed of 3s/token for the ChatGLM-6B model, making private inference of LLMs possible for practical usage.
The limitation of this paper is mainly the non-standard security assumption used for random permutation.
The proposed method has a potential positive societal impact, as it aims to address the data privacy problem in LLMs which is good for the society.
\bibliographystyle{plain}
\bibliography{refs}
\end{document}